\documentclass[11pt]{article}

\usepackage{geometry}                % See geometry.pdf to learn the layout options. There are lots.
\usepackage{adjustbox}
\usepackage{booktabs}
\usepackage[table]{xcolor}
\usepackage{siunitx}
\usepackage{amsmath}
\usepackage{array}
\usepackage{graphicx}
\usepackage{amssymb}
\usepackage{epstopdf}
\usepackage{float}
\usepackage[english]{babel}
\usepackage[latin1]{inputenc}
\usepackage{filecontents}
\usepackage{caption}
\usepackage{subfig}
\usepackage{amsthm}
\newtheorem{theorem}{Theorem}

\newtheorem{corollary}{Corollary}
\usepackage{authblk}

\begin{document}

\title{Competent hosts and endemicity of multi-host diseases}

\date{}

\author[1]{Camilo Sanabria}
\author[1]{Esteban Vargas}

\affil[1]{Department of Mathematics, Universidad de los Andes, Bogota, Colombia}

\maketitle

%%%%%%%%%%%%%%%%%%%%%%%%%%%%%%%%%%%%%%%%%%%%%%%%%%%%%%%%%%
%%%%%%%%%%%%%%%%%%%%%%%%%%%%%%%%%%%%%%%%%%%%%%%%%%%%%%%%%%

\begin{abstract}
In this paper we propose a method to study a general vector-hosts mathematical model in order to explain how the changes in biodiversity could influence the dynamics of vector-borne diseases. We find that under the assumption of frequency-dependent transmission, i.e. the assumption that the number of contacts are diluted by the total population of hosts, the presence of a competent host is a necessary condition for the existence of an endemic state. In addition, we obtain that in the case of an endemic disease with a unique competent and resilient host, an increase in its density amplifies the disease.
\end{abstract}

\section{Introduction}

% Some tropical diseases are amplified by one or several reservoirs. This is the case in diseases such as Chagas disease and Leishmaniasis. Indeed, Chagas disease has a domiciliary cycle, where domestic animals act as reservoirs, and a sylvatic cycle, where mammals like rodents are reservoirs \cite{paniker2007textbook}. Regarding Leishmaniasis, the main reservoirs of the disease in countries of South America are dogs, but other mammals could also act as reservoirs. In this paper, we are interested in diseases that have a network of reservoirs. We are also interested in representing those diseases in a simple mathematical model where we can measure the amplification effects of the reservoirs through the basic reproductive number.

% Transmission disease is not density dependent when we have high mobile vectors (\cite{dobson2004population}) or when we have behavioral encounters that keeps those contacts constant in varying densities of the host. An example of this is presented for a rodent disease in \cite{clay2009testing}. However, if transmission is density dependent, the loss of biodiversity could entail increase of competent hosts, i.e. host that are highly  effective in transmitting the disease. 

The abundance of hosts of a vector-borne disease could influence the dilution or amplification of the infection. In \cite{keesing2010impacts}, the authors discusse several examples where loss of biodiversity increases disease transmission. For instance, West Nile virus is a mosquito-transmitted disease and it has been shown that there is a correlation between low bird density and amplification of the disease in humans \cite{allan2009ecological,swaddle2008increased}. One of the suggested explanations of this phenomenon is that the competent hosts persist as biodiversity is lost, meanwhile the density of the species who reduce the pathogen transmission declines.  This is the case of the Lyme disease in North America, which is transmitted by the blacklegged tick \textit{Ixodes pacificus}. The disease has the white-footed mouse \textit{Peromyscus leucopus} as competent host, which are abundant in either low-diversity or high-diversity ecosystems. On the other hand, the opossum \textit{Didelphis virginiana}, which is a suboptimal host and acts as a buffer of the disease, is poor in low-diversity forest \cite{keesing2009hosts,LoGiudice2008impact}.

Symmetrically, the dilution effect hypothesizes that increases in diversity of host species may decrease disease transmission \cite{ostfeld2000biodiversity}. The diluting effect of the individual and collective addition of suboptimal hosts is discussed in \cite{johnson2010diversity}. For example, the transmission of   \textit{Schistosoma mansoni} to target snail hosts \textit{Biomphalaria glabrata} is diluted by the inclusion of decoy hosts. These decoy hosts are individually effective to dilute the infection. However, it is interesting to notice that their combined effects are less than additive \cite{johnson2012parasite,johnson2009community}.

The objective of this paper is to study the behavior of a vector-borne disease with multiple hosts when changes in biodiversity occur.  More precisely, we present a mathematical framework that simultaneously explains why the accumulative effect of decoy hosts is less than additive and how competent and resilient host amplify the disease. To model a vector-borne disease with multiple hosts we use a dynamical system that was created based on \cite{dobson2004population}. We suggest a mathematical interpretation of competent and suboptimal host using the basic reproductive number of the cycle formed by the host and the vector. Furthermore, we assume that the abundances of the hosts follow a conservation law given by community constraints and with it we attempt to capture how a disturbance of the ecosystem leads to changes in the density of the hosts. We also give a mathematical interpretation of what a resilient species is using the conservation law. In this way, we are able to measure the effect on the dynamics of the disease due to different changes in the biodiversity. We show that in the case of endemic diseases these effects are determined by the effectiveness of the hosts to transmit the disease and the resistance of the hosts to biodiversity changes.
% In mathematical models of infectious diseases based on ordinary differential equations, the basic reproductive number of the disease is frequently obtained using the method of the Next Generation Matrix (NGM) presented in \cite{van2002reproduction}. Different interpretations of the NGM can lead to different basic reproductive numbers. In Section \ref{sappendix} we present the construction of the NGM that it is used in this work. 

In section \ref{smodel} we present the variables and the equations of the model. Section \ref{sresults} is divided in three subsections. In subsection \ref{scompetent} we derive some properties of the basic reproductive number and we show how an endemic state implies the existence of a competent host. From these properties we explain why the combined effect of decoy hosts is less than additive and how biodiversity loss can entail amplification of the disease. Subsection \ref{sconstraints} introduces the community constraints that leads us to a definition of resilient host. In subsection \ref{singlecompetent} we consider the case of an endemic disease with a unique competent host. We discuss the conclusions from our results in section \ref{sconclusions} . The mathematical justification are in Appendix, section \ref{sappendix}.

%%%%%%%%%%%%%%%%%%%%%%%%%%%%%%%%%%%%%%%%%%%%%%%%%%%%%%%%%%
%%%%%%%%%%%%%%%%%%%%%%%%%%%%%%%%%%%%%%%%%%%%%%%%%%%%%%%%%%

\section{The model} \label{smodel}

We propose a mathematical model of a vector-borne disease that is spread among a vector $V$ and hosts $H_i$, $i= 1, \ldots,k$.  We suppose that each population is divided into susceptible individuals ($S_V$ susceptible vectors and $S_{H_i}$ susceptible hosts) and infectious individuals ($I_V$ infectious vectors and $I_{H_i}$ infectious hosts). Let $N_V$ and $N_{H_i}$ represent the total abundances of  vectors and hosts respectively. The dynamics of the disease will be studied by means of the basic reproductive number as we are interested in the strength of a pathogen to spread in an ecosystem. Modification of the ecosystem entails changes in the abundances of the hosts. After these changes are brought, the ecosystem will settle to a stable pattern of constant abundances. We are interested in understanding the basic reproductive number when the ecosystem reaches these steady states. Therefore we will assume the abundance of the vector and hosts are constant in time, i.e. $\dot{N_V} = \dot{S_V} + \dot{I_V}=0$ and $\dot{N_{H_i}} = \dot{S_{H_i}} + \dot{I_{H_i}}=0$ for $i=1,\ldots,k$. In that way, it suffices to consider as state variables  only the number of infectious species. We define the total number of hosts as $N_H= \sum_{i=1}^{k}N_{H_i}$. Our model is a system of ordinary differential equations for the infectious populations of hosts and vectors:
\begin{equation}
\begin{cases}
\dot{I_{H_i}} = \beta_{VH_i} I_V \dfrac{S_{H_i}}{N_{H}}  - \delta_{H_i} I_{H_i},\quad i = 1, \ldots, k\\
{}\\
\dot{I_V} =    \sum_{i=1}^k \beta_{H_iV} I_{H_i}\dfrac{N_{H_i} }{N_H}\dfrac{S_V}{N_V} - \delta_V I_V.\\
\end{cases}
\label{ecompleto}
\end{equation}
We assume frequency-dependent transmission and that the vector does not have preference for a specific host, hence the number of contacts between the vector and the hosts are diluted by the total population of hosts. We also assume that there are no intraspecies infections and that there is no interspecies infection between hosts, or that these are negligible. Therefore,  the only mean of infection is through contact with the vectors as Fig. \ref{complete} shows.

\begin{figure}[!hbp]
\centering
\includegraphics[scale = 0.5]{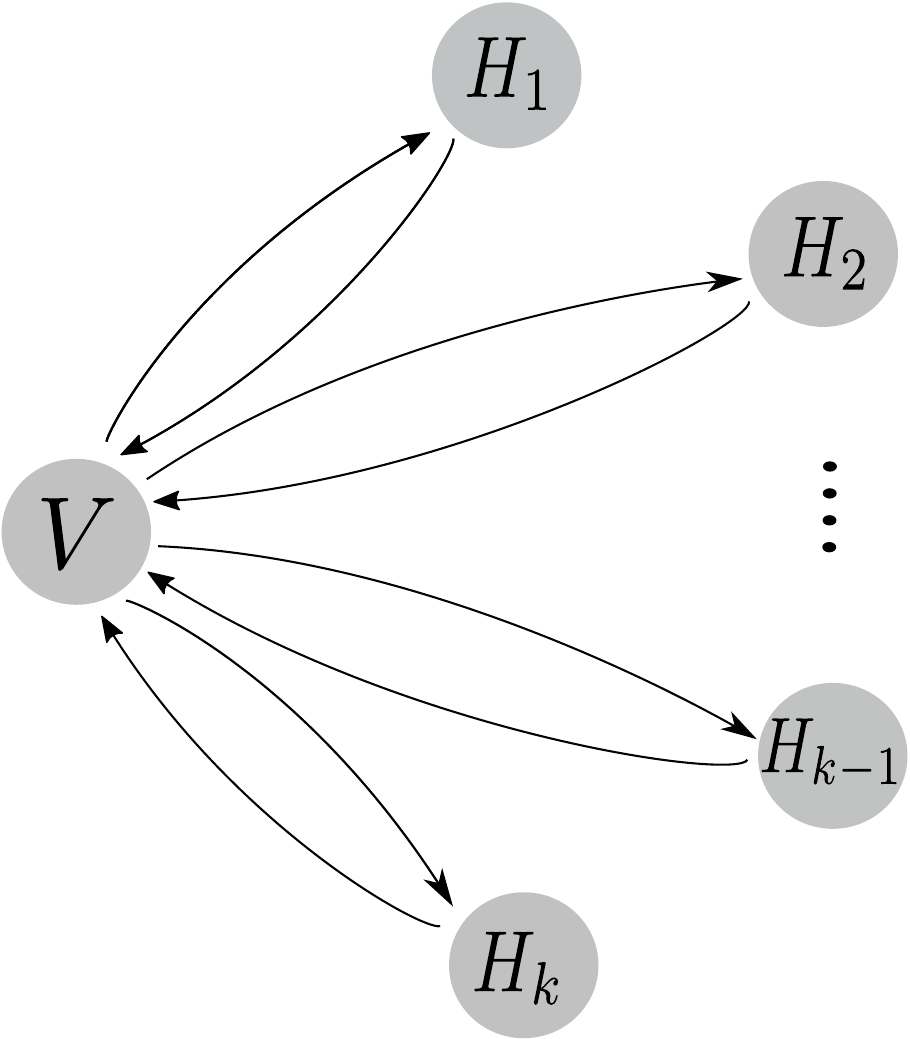}
\caption{The node $V$ represents the infectious vector and the nodes $H_i, i=1,\ldots,k $ represent the infectious reservoirs.}
\label{complete}
\end{figure}

The parameters of the model are presented in Table \ref{param1}.

Note that we could alternatively assume that infected hosts gain immunity after recovering. In such case the model would yield the same next generation matrix (see Appendix \ref{ssngm}), and since our analysis depends entirely on this matrix we would obtain the same results.

\begin{table}
\centering
\begin{tabular}{llll}
\hline\noalign{\smallskip}
Parameter & Definition & Units \\ 
\noalign{\smallskip}\hline\noalign{\smallskip}

$\beta_{VH_i}$ & Transmission rate from $V$ to $H_i$ & $[H_i]/([time]*[V])$ \\
& in the cycle formed by $V$ and $H_i$ & &\\
{}\\
$\beta_{H_iV}$ & Transmission rate from $H_i$ to $V$ & $[V]/([time]*[H_i])$\\
& in the cycle formed by $V$ and $H_i$ & \\
{}\\
$\delta_V$ & Mortality rate of infected vectors & $1/[time]$ \\
{}\\
$\delta_{H_i}$ & Mortality rate of infected hosts $H_i$ & $1/[time]$ \\

\noalign{\smallskip}\hline
\end{tabular}
\caption{Parameters of the model described by equations (\ref{ecompleto}).}\label{param1}
\end{table}

\section{Results} \label{sresults}

\subsection{Properties of the basic reproductive number and the existence of competent hosts}\label{scompetent}

We define the basic reproductive number $\mathcal{R}_0^{H_i}$ of the cycle formed by host $H_i$ and the vector $V$ by
\begin{equation*}
(\mathcal{R}_0^{H_i})^2= \dfrac{\beta_{VH_i}}{\delta_V}\dfrac{\beta_{H_iV}}{\delta_{H_i}}.
\end{equation*}
The quantity $\mathcal{R}_0^{H_i}$ is the basic reproductive number of the epidemiological model (\ref{ecompleto}) when $N_H=N_{H_i}$. It corresponds to the average number of secondary cases produced by a single infected host $H_i$ in an otherwise susceptible population when the only cycle taken into account is the interaction between $V$ and $H_i$. In this setting, the infection will spread in the population if $\mathcal{R}_0^{H_i} > 1$, and it will disappear if $\mathcal{R}_0^{H_i}<1$. Therefore, we say that a host $H_i$ is competent if $\mathcal{R}_0^{H_i} \geq 1$ and suboptimal if $\mathcal{R}_0^{H_i}<1$.

In general, taking into account all cycles, if $D_i = \frac{N_{H_i}}{N_H}$ is the density of the host $H_i$ in the total population of hosts, then the basic reproductive number $\mathcal{R}_0$ of the whole system is given by
\begin{equation*}
\mathcal{R}_0^2 = \sum_{i=1}^k (\mathcal{R}_0^{H_i})^2 D_i^2,
\end{equation*}
(see (\ref{r0msimple}) in Appendix). Note that this implies that the combined effect of decoy hosts is less than additive.
% The contour plots of (\ref{r0dobson}) are also hyperplanes. We also have that 
% \begin{equation}\label{partial}
% \frac{\partial \mathcal{R}_0}{\partial N_i} = \frac{1}{N_t \mathcal{R}_0} (\rho_i D_i - \mathcal{R}_0^2) 
% \end{equation}

The quantity $\mathcal{R}_0$ is a convex function of $D_1, \ldots D_{k}$. We have $D_i \geq 0$ for $i = 1, \ldots, k$ and $\sum_{i=1}^kD_i = 1$. Using Lagrange multipliers, we obtain that the minimum value of $\mathcal{R}_0$ is attained in $(D_1^* , \ldots, D_k^*)$, where $$ (\mathcal{R}_0^{H_1})^2D_1^*= \ldots =  (\mathcal{R}_0^{H_k})^2D_k^*.$$ Therefore, we have
$$ D_i^* = \dfrac{\frac{1}{ (\mathcal{R}_0^{H_i})^2}}{\sum_{j=1}^k \frac{1}{ (\mathcal{R}_0^{H_j})^2}}\quad\textrm{ for } i=1,\ldots,k $$
and 
\begin{equation}\label{eharmonic}
 (\mathcal{R}_0)_{\min}^2 = \dfrac{1}{\sum_{j=1}^k \frac{1}{ (\mathcal{R}_0^{H_j})^2}}= \dfrac{1}{k}H\left((\mathcal{R}_0^{H_1})^2, \ldots, (\mathcal{R}_0^{H_k})^2\right),
\end{equation}
where $H\left((\mathcal{R}_0^{H_1})^2, \ldots, (\mathcal{R}_0^{H_k})^2\right)$ is the harmonic mean of $(\mathcal{R}_0^{H_1})^2, \ldots, (\mathcal{R}_0^{H_k})^2$. From the properties of the harmonic mean we have  
$$ \underset{i=1, \ldots, k}{\min}\{(\mathcal{R}_0^{H_i})^2\} \leq H\left((\mathcal{R}_0^{H_1})^2, \ldots, (\mathcal{R}_0^{H_k})^2\right) \leq k\underset{i=1, \ldots, k}{\min}\{(\mathcal{R}_0^{H_i})^2\}.$$
Using (\ref{eharmonic}), we obtain
\begin{equation*}
 \frac{1}{k}\underset{i=1, \ldots, k}{\min}\{(\mathcal{R}_0^{H_i})^2\} \leq (\mathcal{R}_0)_{\min}^2 \leq \underset{i=1, \ldots, k}{\min}\{(\mathcal{R}_0^{H_i})^2\}.
\end{equation*}
From the last inequalities we can observe the following. First, the presence of a reservoir with $\mathcal{R}_0^{H_i}<1$ implies that $ (\mathcal{R}_0)_{\min}<1$. Hence, in some cases we may have $\mathcal{R}_0<1$. Furthermore, from (\ref{eharmonic}) we obtain that the larger the number of the hosts is, the smaller the basic reproductive number could be. This explains how high biodiversity could lead to the dilution of the disease. On the other hand, if all the reservoirs are effectively transmitting the disease  ($\mathcal{R}_0^{H_i}\gg 1, i = 1, \ldots, k$) and there are few host ($k$ is small), then $\mathcal{R}_0>1$. This explains why in the case when competent host species thrive as a result of biodiversity loss we can expect the amplication of the disease, as discussed in \cite{keesing2010impacts} for the case of the Lyme disease \cite{keesing2009hosts,LoGiudice2008impact} and the Nipah virus \cite{epstein2006nipah}.

Furthermore, as the function $(\mathcal{R}_0)^2(D_1, \ldots, D_k)$ is convex, we have
\begin{equation*}
\mathcal{R}_0^2 \leq \underset{i=1, \ldots, k}{\max}\{(\mathcal{R}_0^{H_i})^2\}.
\end{equation*}
This inequality implies that the disease can not be amplified beyond the basic reproductive number of the most competent host. We obtain the following theorem.

\begin{theorem} \label{tcompetent}
There exist values of $D_1, \ldots, D_k$ for which $\mathcal{R}_0\ge 1$ if and only if $$ (\mathcal{R}_0)_{\min} < 1 < \mathcal{R}_0^{H_i},$$ for some $i$. In particular, under the assumption of model (\ref{ecompleto}), the endemicity of a disease implies the existence of a competent host.
\end{theorem}

Figure \ref{dobsonf} represents a contour plot of $\mathcal{R}_0 $ in the case of two hosts. 

\begin{figure}[H]
\centering
\includegraphics[scale = 0.3]{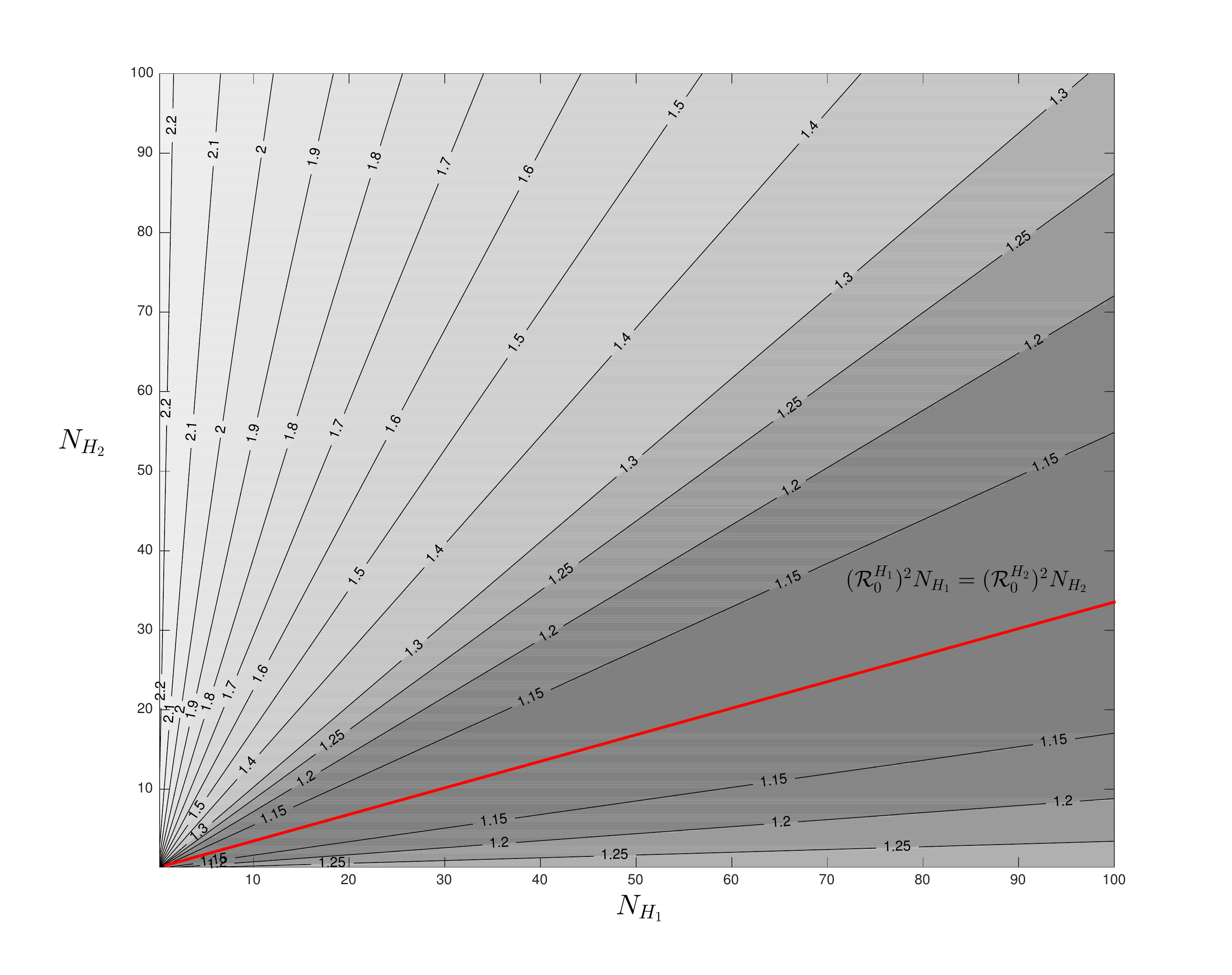}
\caption{Contour plot for different values of $\mathcal{R}_0 $ for a two hosts system where there  is a competent host (horizontal axis)  and a suboptimal host (vertical axis). In the red line $\mathcal{R}_0 $ takes its minimum value, and as we move away from the red line, $\mathcal{R}_0 $ increases.}
\label{dobsonf}
\end{figure}

%%%%%%%%%%%%%%%%%%%%%%%%%%%%%%%%%%%%%%%%%
%%%%%%%%%%%%%%%%%%%%%%%%%%%%%%%%%%%%%%%%%

\subsection{Community constraints}\label{sconstraints}

In this section we will take into consideration host interaction using community constraints. First, we will consider the case when the abundance of hosts follow linear constraints. Secondly, we will show that in the study of small changes in the abundances we can linearize general constraints.

\subsubsection{Linear case}

Let us assume that the abundances of the hosts $N_{H_1},\ldots, N_{H_k}$ follow $k-1$ linear constraints: 
\begin{equation*}
\sum_{j=1}^k a_{ij} N_{H_i} + b_i = 0,\quad\textrm{ for } i =1,\ldots, k-1,
\end{equation*}
for some constants $a_{ij}$, $b_i$.

If the matrix $(a_{ij})_{1\le i,j\le k-1}$ is nonsingular, the abundance of all hosts can be explained by the abundance of the host $H_k$: 

\begin{equation}\label{linconstraints}
N_{H_i} = -A_i N_{H_k} + B_i, \quad\textrm{ for } i=1\ldots, k-1,
\end{equation}
for some constants $A_i$, $B_i$. In particular, if $A_i > 0$ in (\ref{linconstraints}), then $N_{H_k}$ increases as $N_{H_i}$ decreases. Moreover, when $A_i=\dfrac{dN_{H_i}}{dN_{H_k}}>1$ the changes in $N_{H_i}$ are more pronounced than the changes in $N_{H_k}$. Therefore, we say that the host $k$ is the resilient if $A_i>1$ for $i=1,\ldots, k-1$ and it is non-resilient if $0<A_i < 1$ for $i=1,\ldots, k-1$.

We have
\begin{equation*}
\frac{d \mathcal{R}_0}{d N_{H_k}}=D_{\mathbf{u}}\mathcal{R}_0 = {\sum_{i=i}^k u_i r_i},
\end{equation*}
where $\mathbf{u}=(-A_1, \ldots, - A_{k-1},1 )$ and  $r_i=\dfrac{\partial \mathcal{R}_0}{\partial N_{H_i}}= \dfrac{1}{N_H \mathcal{R}_0} \left((\mathcal{R}_0^{H_i})^2 D_i - \mathcal{R}_0^2\right) $.\\

We define the index 
\begin{equation*}
\Gamma_k =  \frac{N_{H_k}}{\mathcal{R}_0} \frac{d \mathcal{R}_0}{d N_{H_k}}
\end{equation*}
The index $\Gamma_k$ measures the sensitivity of $\mathcal{R}_0$ to changes of the population $N_k$.

\subsubsection{General constraints}\label{Gc}

Let us assume that the abundances of the hosts $\mathbf{N}=(N_{H_1},\ldots, N_{H_k})$ follow the $m$ community constraints: 
\begin{equation*}
\mathbf{F}(\mathbf{N}) = (F_1(\mathbf{N}), \ldots, F_m(\mathbf{N})) = (0, \ldots, 0)=\mathbf{0},
\end{equation*}
for some $m < k$. Here $F_1,\ldots,F_m$ are real-valued differentiable functions defined where the values for $\mathbf{N}$ have biological sense. Let $E$ be the set of such values of $\mathbf{N}$ where the community constraints are satisfied and let $\mathbf{N}_0\in E$. Under suitable conditions (see subsection \ref{jd} in Appendix), we have $$N_i = g_i(N_{m+1}, \ldots , N_{k})\quad\textrm{ for } i=1,\ldots,m,$$ for some functions $g_1, \ldots, g_m$ and for $\mathbf{N}\in E$ close to $\mathbf{N}_0$. The derivatives $\dfrac{\partial g_i}{\partial N_{j}}$, $i = 1, \ldots, m$, $j=m+1,\ldots,k$ can be computed in terms of the derivatives of the functions $F_1, \ldots, F_m$.

If $m=k-1$ and $\dfrac{\partial \mathcal{R}_0}{\partial N_k}(\mathbf{N}_0)\ne 0$ in a neighborhood of $\mathbf{N}_0$, then we have $$N_i = g_i(N_k)\quad\textrm{ for } i=1, \ldots, k-1.$$ 
%$(N_k^0, N_i^0)$
Moreover, for all $\mathbf{N}\in E$ close to $\mathbf{N}_0$ we have the approximation 
$$N_i = g_i(N_k) \approx  -A_i N_k + B_i,$$
for some constants $A_i$, $B_i $ (see Appendix). Thus, locally we can consider linear restrictions as in (\ref{linconstraints}).
% As an example, let us assume that $k=1, m=2$ and $$ F_1 = a_1 x + b_1 y + z + c1, F_1 = a_2 x + b_2 y + z + c2 $$ 

% Let us also assume that $J_d = \begin{pmatrix}
% a_1 & b_1 \\
% a_2 & b_2 
% \end{pmatrix}$ is non-singular. Therefore, each $\overline{u}=(u_1,u_2,u_3)\in Ker(J)$ can be written as function of $u_3$ given that 

% $$\begin{pmatrix}
% u_1 \\
% u_2
% \end{pmatrix} = u_3 J_d^{-1}J_i = u_3 \begin{pmatrix}
% a_1 & b_1 \\
% a_2 & b_2 
% \end{pmatrix} ^{-1} \begin{pmatrix}
% 1 \\
% 1
% \end{pmatrix}= \frac{u_3}{a_1b_2-a_2b_1} \begin{pmatrix}
% b_2 - b_1 \\
% a_1-a_2
% \end{pmatrix} $$

%%%%%%%%%%%%%%%%%%%%%%%%%%%%%%%%%%%%%%%%%%%%%%%%%%%%%%%%%%
%%%%%%%%%%%%%%%%%%%%%%%%%%%%%%%%%%%%%%%%%%%%%%%%%%%%%%%%%%

\subsection{The case of a single competent host}\label{singlecompetent}

In this section we consider the case of an endemic disease. Theorem \ref{tcompetent} implies the existence of a competent host in this setting. We will show that in the case when this competent host is unique the increase in its density implies the amplification of the disease if the densities of the rest of the hosts decrease. This corresponds to the cases when there is a unique host that thrives with biodiversity loss and this host is competent. 

\begin{theorem}\label{t1competent}
We assume that $\mathcal{R}_0^{H_i}<1$ for $i =1, \ldots, k-1$ and $\mathcal{R}_0^{H_k}>1$. Let $D_1, \ldots,D_k $ be such that $\mathcal{R}_0 \geq 1$. Then $$\dfrac{\partial R_0}{\partial D_i} < 0\quad\textrm{ for } i=1,\ldots,k-1$$ and $$\dfrac{\partial R_0}{\partial D_k} > 0.$$
In particular, under the assumption of model (\ref{ecompleto}), in the case of an endemic disease with a unique competent host, increase in its density together with decrease in the density of all other hosts implies amplification of the disease.
\end{theorem}

\begin{proof}
See section \ref{onecompetent} in Appendix.
\end{proof}

\begin{corollary}
Under the assumption of model (\ref{ecompleto}), in the case of an endemic disease with a unique competent and resilient host, increase in its density implies amplification of the disease.
\end{corollary}

Let us assume 
\begin{equation*}
N_{H_i} = -A N_{H_k} + B_i \quad\textrm{ for } i = 1,\ldots, k-1.
\end{equation*}
for some constants $A$, $B_i$. If $A>1$, then the host $H_k$ is resilient and, the greater $A$ is, the more resilient $H_k$ is. We have that $\Gamma_k$ is an increasing function of $A$ (see subsection \ref{jd} in Appendix). Furthermore, taking $D_k$ and $A$ large, we have 
\begin{equation}\label{egammak}
\Gamma_k \approx  (k-1)A.
\end{equation}
Hence $ \Gamma_k $ increases as $k$, $D_k$ and $A$ increase. This implies that the more resilient the host $H_k$ is, the greater its effect on $\mathcal{R}_0$ is, in the case when this host is abundant. The case $k=2$ is represented in Fig. \ref{const2}. 

\begin{figure}[H]
\subfloat[]{\includegraphics[scale=0.325]{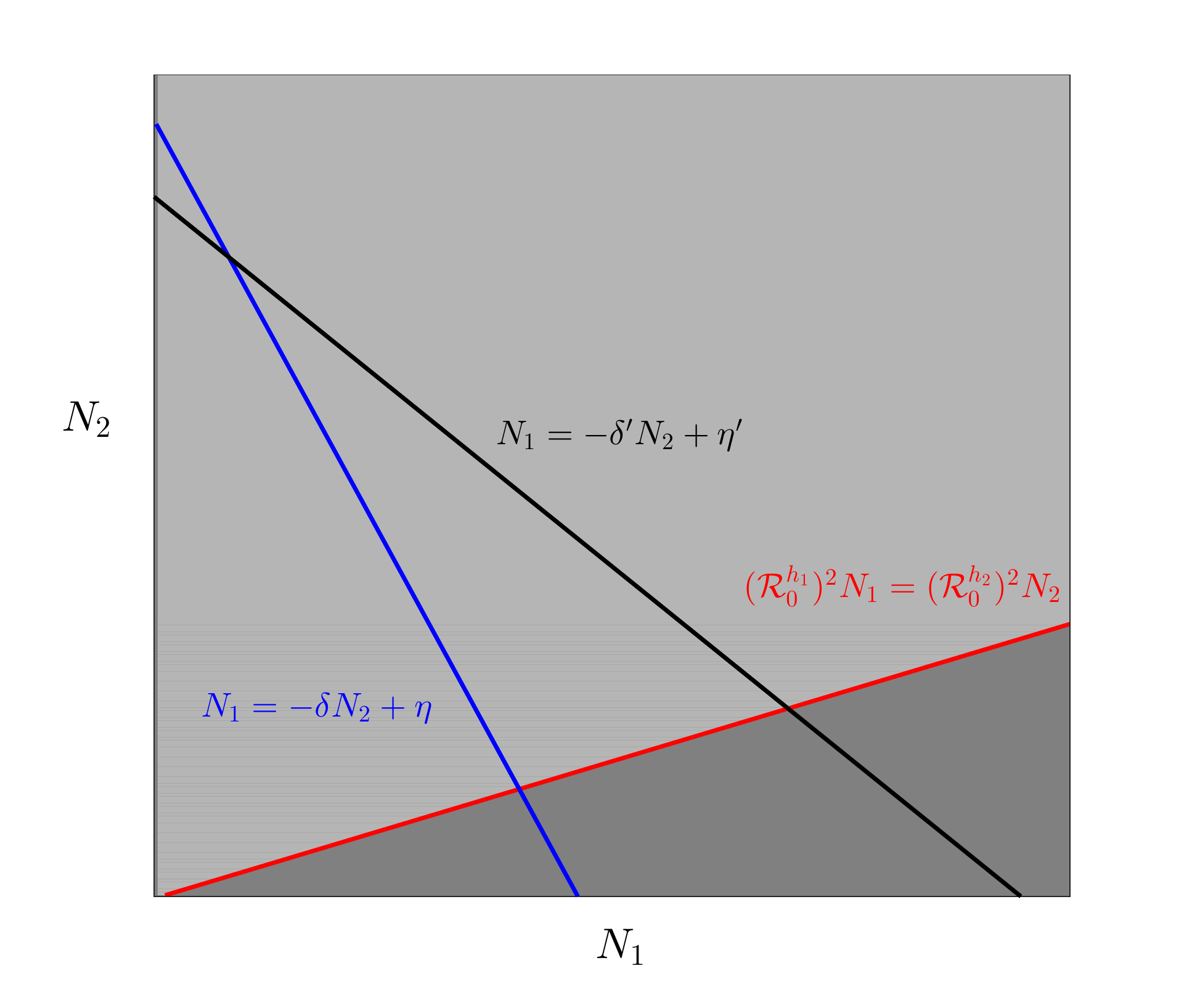}}
\subfloat[]{\includegraphics[scale=0.25]{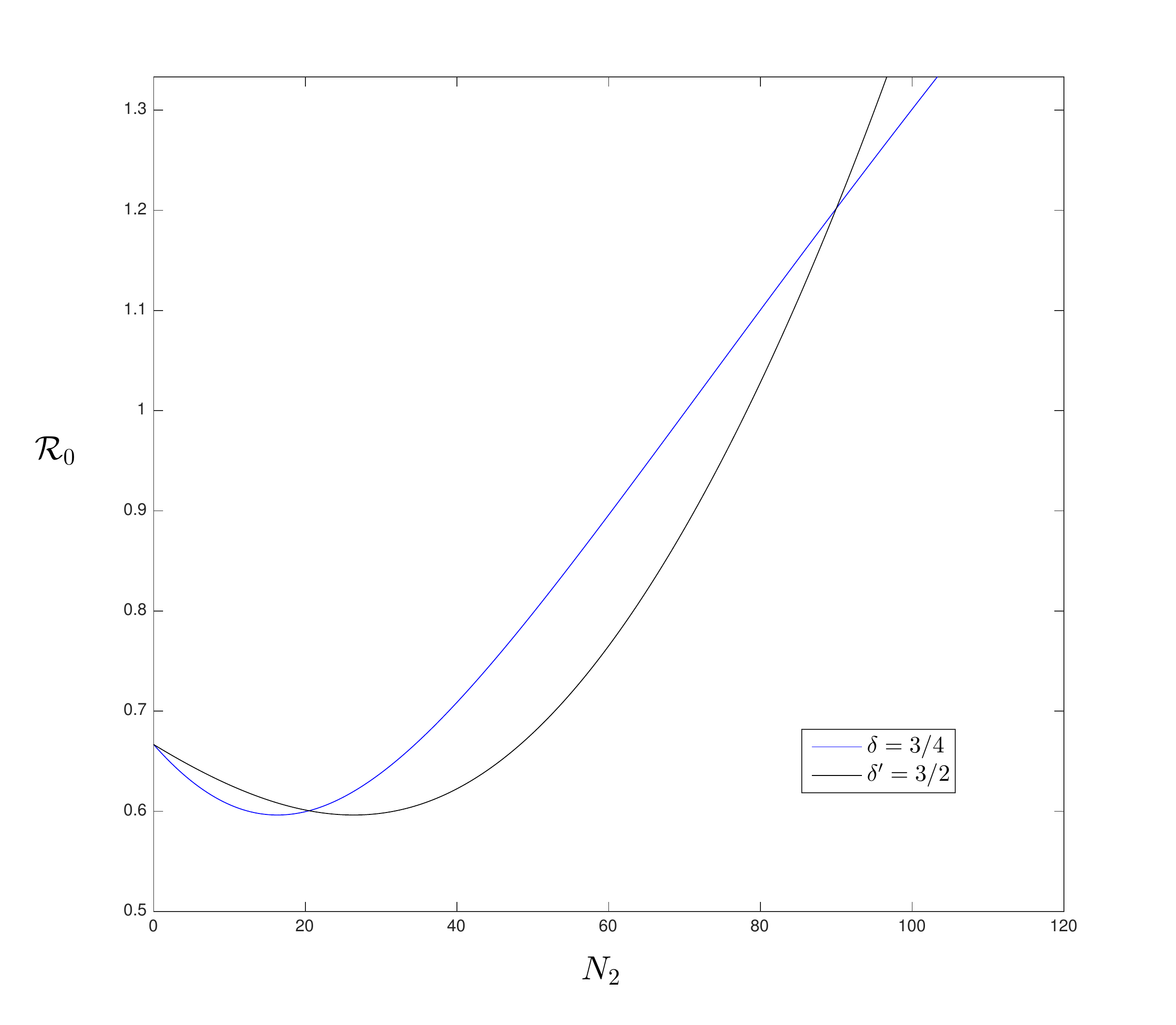}}\\
\centering{\subfloat[]{\includegraphics[scale=0.25]{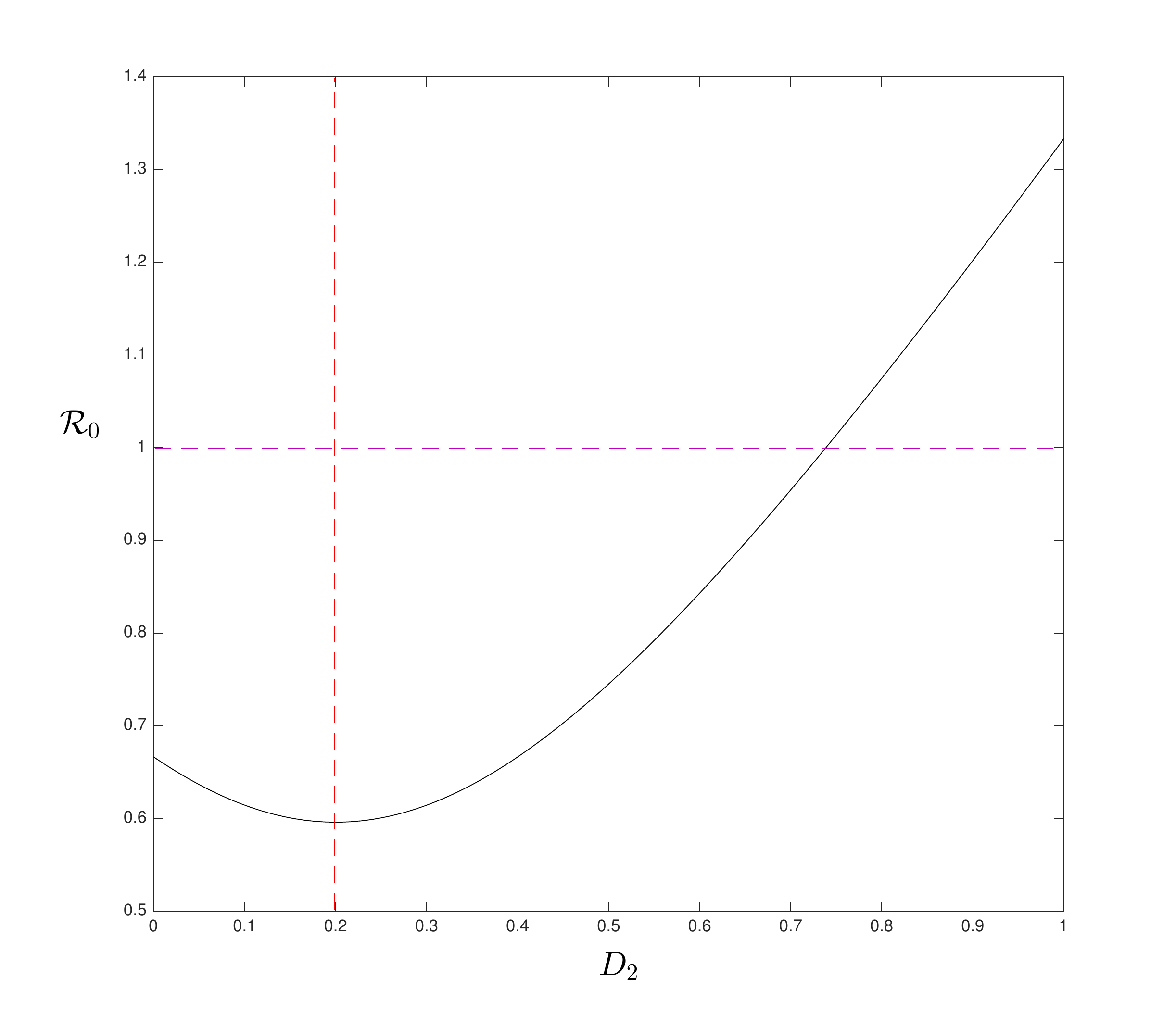}}}

\caption{In figure (a) the host $H_2$ is the competent host (the slope of the red line is less than one). The blue and the black lines represent the community linear constraints. Over the blue line the host $H_2$ is non-resilient, whereas over the black line the host $H_2$ is resilient. In figure (b) the blue graph represents $\mathcal{R}_0$ when host $H_2$ is competent and non-resilient and the black graph represents the case when $H_2$ is competent and resilient. Close to the intersection of these graphs (where $ D_2 = 0.9$ and $ D_1=0.1$) the derivative $\dfrac{d\mathcal{R}_0}{dN_{H_2}}$ of the black graph is greater than of the blue one. Moreover, $\Gamma_2$ is greater for the black graph than for the blue one, as expected by (\ref{egammak}). In this simulation $\mathcal{R}_0^{H_1} = 2/3$, $\mathcal{R}_0^{H_2} = 4/3$. In figure (c), on the left side of the dashed red line (where the minimum of $\mathcal{R}_0$ is attained), there is dilution of the disease if $D_2$ increases. On the right side of the red line, there is amplification of the disease as $D_2$ increases. Moreover, the amplification of the disease above the dashed purple line (where $\mathcal{R}_0\geq 1$) is ensured by Theorem \ref{t1competent}.}
\label{const2}
\end{figure}

\section{Conclusions} \label{sconclusions}

In this paper we present a mathematical framework that explains how changes of biodiversity can lead to the dilution or amplification of the disease. We show that the square of the basic reproductive number of the whole ecosystem is the weighted average of the squares of the basic reproductive numbers of the cycles between the vector and the hosts, weighted by their densities. Therefore, the accumulative effect of the hosts that buffer the disease is less than additive. Moreover, we obtain that the mininum of the basic reproductive number of the whole system is the harmonic mean of the basic reproductive numbers of the cycles. Hence, we conclude that an increase in biodiversity could dilute the disease and that loss in biodiversity could amplify the disease. Furthermore, we obtain that a necessary condition for the endemicity of a disease is the presence of a competent host. 

Finally, we study the case of an endemic disease. To explain how changes in the ecosystem affects the density of the hosts we assume that the abundances of the hosts follow a conservation law given by community constraints. We show that in the case when we have small changes in abundances, general constraints can always be linearized, thus it is sufficient to consider only linear constraints. We obtain that in the case of a disease with a unique resilient and competent host increase in its density amplifies the infection.

\section{Appendix}\label{sappendix}

\subsection{Next generation matrix} \label{ssngm}

We will compute $\mathcal{R}_0$ using the NGM method from \cite{van2002reproduction}.
From model (\ref{ecompleto}) we obtain the matrices $F$ and $V$ that define the NGM:

$$F=
\begin{pmatrix}

 0 &  \beta_{H_1V} D_1  & \beta_{H_2V} D_2 & \ldots & \beta_{H_kV} D_k \\
 \beta_{VH_1} D_1 & 0 & 0 & 0 & 0 \\
 \beta_{VH_2} D_2& 0 & 0 & 0 & 0 \\
 \ddots & \ddots & \ddots & \ldots &0\\
 \beta_{VH_k} D_k& 0 & 0 & 0 & 0 \\

\end{pmatrix}, V=
\begin{pmatrix}
  \delta_V &  0 & 0 & 0 & 0 \\
 0 & \delta_{H_1} & 0 &  0 & 0 \\
 0 & 0 & \delta_{H_2} & \ldots & 0 \\
 \ddots & \ddots & \ddots & \ldots &0\\
 0 & 0 & 0 & \ldots & \delta_{H_k} \\

\end{pmatrix}.$$
Hence, the NGM is:

\begin{equation*}
G = FV^{-1}=
\begin{pmatrix}
  0 &  \frac{\beta_{H_1V}}{\delta_{H_1}} D_1 & \frac{\beta_{H_2V}}{\delta_{H_2}} D_2 & \ldots & \frac{\beta_{H_kv}}{\delta_{H_k}} D_k \\
 \frac{\beta_{VH_1}}{\delta_{V}} D_1 & 0 &0& \ldots & 0 \\
 \frac{\beta_{VH_2}}{\delta_{V}} D_2& 0 &0& \ldots & 0 \\
 \vdots & \ddots & \ddots & \ldots &0\\
 \frac{\beta_{VH_k}}{\delta_{V}} D_k&  0 &0& \ldots & 0 \\\end{pmatrix}.
\end{equation*}
Computing the spectral radius of the matrix $G$, we obtain that the basic reproductive number of the whole system is given by
\begin{equation} \label{r0msimple}
 \mathcal{R}_0=\rho(FV^{-1}) = \sqrt{\sum_{i=1}^k \frac{\beta_{VH_i}}{\delta_V} \frac{\beta_{H_iV}}{\delta_{H_i}}D_i^2}.
\end{equation}

The disease free equilibrium (DFE) of model (\ref{ecompleto}) is $\mathbf{I}^*=(I_V^*,I_{H_1}^*,\ldots,I_{H_k}^*)=\mathbf{0}$. The following theorem explains how the basic reproductive number is related to the stability of the DFE in model (\ref{ecompleto}) \cite[Theorem 2]{van2002reproduction}.

\begin{theorem}
Let $\mathbf{I}^*$ be the DFE of (\ref{ecompleto}). Then, $\mathcal{R}_0<1$ implies that $\mathbf{I}^*$ is locally asymptotically stable and $\mathcal{R}_0>1$ implies that $\mathbf{I}^*$ is unstable.  
\label{umbral}
\end{theorem}

\subsection{Community constraints}\label{jd}

Let $F_1,\ldots,F_m$ and $E$ be as in subsection \ref{Gc} and let $\mathbf{N}_0\in E$. We assume that the matrix 
\begin{equation*}
J_1 = \frac{\partial(F_1, \ldots, F_m)}{\partial(N_{H_1}, \ldots, N_{H_m})}(\mathbf{N}_0) = \left( \frac{\partial F_i (\mathbf{N}_0)}{\partial N_{H_j}}\right)_{1\le i,j\le m}
\end{equation*} 
is invertible and let us define 
\begin{equation*}\label{ji}
J_2 = \frac{\partial(F_1, \ldots, F_m)}{\partial(N_{H_{m+1}}, \ldots, N_{H_m})}(\mathbf{N}_0) = \left( \frac{\partial F_i (\mathbf{N}_0)}{\partial N_{H_j}}\right)_{1\le i\le m,m< j\le k}.
\end{equation*}

The implicit function theorem states that there exists a neighborhood in $E$ of $\mathbf{N}_0$ where we have $N_i = g_i(N_{m+1}, \ldots , N_{k})$ for $i = 1, \ldots, m$. Furthermore, if $\mathbf{g}= (g_1, \ldots, g_m)$, then 
$$\frac{\partial\mathbf{g}}{\partial N_{H_j}} =
\begin{pmatrix}
\frac{\partial g_1}{\partial N_{{H_j}}} \\
\vdots \\
\frac{\partial g_m}{\partial N_{H_j}}
\end{pmatrix} = - J_1 ^{-1} \begin{pmatrix}
\frac{\partial F_1}{\partial N_{H_j}} \\
\vdots \\
\frac{\partial F_m}{\partial N_{H_j}}
\end{pmatrix},$$
for $m< j\le k$. 

We define $$ J = \frac{\partial(F_1, \ldots, F_m)}{\partial(N_{H_1}, \ldots, N_{H_k})}(\mathbf{N}_0) = \left( \frac{\partial F_i (\mathbf{N}_0)}{\partial N_{H_j}}\right) _{1\le i\le m, 1\le j \le k}.$$
We are interested in computing $D_{\mathbf{u}} \mathcal{R}_0$ for $\mathbf{u}\in T_{\mathbf{N}_0} E$, where $$ T_{\mathbf{N}_0} E = \{\mathbf{u} \in \mathcal{R}^k | J \mathbf{u} = \mathbf{0} \}.$$
If $\mathbf{u} = (u_1, \ldots, u_k)$, using $u_{m+1}, \ldots, u_k$ as free variables, we have that 
$$
\begin{pmatrix}
u_1 \\
\vdots \\
u_m
\end{pmatrix} = - \sum_{j=m+1}^k u_j J_1 ^{-1} \begin{pmatrix}
\frac{\partial F_1}{\partial N_{H_j}} \\
\vdots \\
\frac{\partial F_m}{\partial N_{H_j}}
\end{pmatrix} = - J_1^{-1}J_2\begin{pmatrix}
u_{m+1} \\
\vdots \\
u_k
\end{pmatrix}.$$
Therefore, for a given set of values $u_{m+1}, \ldots, u_k$, we can obtain the values $u_1, \ldots, u_m$ and $$D_{\mathbf{u}}\mathcal{R}_0 = \sum_{i=i}^k u_i r_i,$$
where $r_i= \dfrac{1}{N_H \mathcal{R}_0} \left((\mathcal{R}_0^{H_i})^2 D_i - \mathcal{R}_0^2\right)$ is evaluated in $\mathbf{N}_0$.

If we assume $m=k-1$, then there exists a neighborhood in $E$ of $\mathbf{N}_0$ where $$ N_{H_i} = g_i(N_{H_k})\quad\textrm{ for } i=1, \ldots, k-1.$$ 
Furthermore, $$J_2=\begin{pmatrix}
\frac{\partial F_1(\mathbf{N}_0)}{\partial N_{H_j}} \\
\vdots \\
\frac{\partial  F_{k-1}(\mathbf{N}_0)}{\partial N_{H_j}}
\end{pmatrix}
$$
and
$$
\begin{pmatrix}
u_1 \\
\vdots \\
u_{k-1}
\end{pmatrix} = - u_k J_1 ^{-1} \begin{pmatrix}
\frac{\partial F_1(\mathbf{N}_0)}{\partial N_{H_j}} \\
\vdots \\
\frac{\partial  F_{k-1}(\mathbf{N}_0)}{\partial N_{H_j}}
\end{pmatrix} = -  J_1^{-1}J_2u_k,$$
for $(u_1, \ldots, u_{k-1}, u_k)\in T_{\mathbf{N}_0} E$.
Taking $u_k =1$, we have
$$\begin{pmatrix}
u_1 \\
\vdots \\
u_{k-1}
\end{pmatrix} = \begin{pmatrix}
\frac{\partial g_1}{\partial N_{H_k}} \\
\vdots \\
\frac{\partial g_m}{\partial N_{H_k}}
\end{pmatrix}.
$$
Therefore, for $\mathbf{N}\in E$ close to $\mathbf{N}_0$ we have the approximations 
$$N_{H_i} = g_i(N_{H_k}) \approx u_i (N_{H_k} - N_{H_k}^0) + N_{H_i}^0 = -A_i N_{H_k} + B_i,$$
for $i = 1, \ldots, k-1$, where $A_i= -u_i$ and $B_i = N_{H_i}^0 - u_i N_{H_k}^0$. 

\subsection{One competent host}\label{onecompetent}
We assume that $\mathcal{R}_0^{H_i}<1$ for $i =1, \ldots, k-1$ and $\mathcal{R}_0^{H_k}>1$. Let $D_1, \ldots,D_k $ be such that $\mathcal{R}_0 \geq 1$. We will prove that $\dfrac{\partial\mathcal{R}_0}{\partial D_i} < 0$  for $i=1,\ldots,k-1$ and $\dfrac{\partial\mathcal{R}_0}{\partial D_k} > 0$. Using $\sum_{j=1}^{k}D_j=1$, we have
$$ \dfrac{\partial\mathcal{R}_0}{\partial D_i} = \dfrac{1}{\mathcal{R}_0}\left((\mathcal{R}_0^{H_i})^2D_i- (\mathcal{R}_0^{H_k})^2D_k\right)\quad\textrm{ for }i=1,\ldots,k-1.$$
Furthermore, since $\mathcal{R}_0= \sum_{i=1}^k (\mathcal{R}_0^{H_i})^2 D_i^2$, we obtain $$ D_k\left((\mathcal{R}_0^{H_k})^2 D_k - 1\right) \ge \sum_{i=1}^{k-1} D_i\left(1-(\mathcal{R}_0^{H_i})^2 D_i\right) \ge 0.$$  
Therefore, $$ (\mathcal{R}_0^{H_k})^2 D_k \ge 1,$$ hence $$\dfrac{\partial\mathcal{R}_0}{\partial D_i} < 0\quad\textrm{ for } i=1,\ldots,k-1$$ and $$\dfrac{\partial\mathcal{R}_0}{\partial D_k}=\dfrac{\partial\mathcal{R}_0}{\partial D_1}\dfrac{\partial D_1}{\partial D_k}>0.$$ \\  

We have $$ \Gamma_k = \frac{D_k}{\mathcal{R}_0^2}\sum_{i=i}^k u_i ((\mathcal{R}_0^{H_i})^2 D_i - \mathcal{R}_0^2).$$ 
If $\mathbf{u} = (-A, \ldots, -A, 1)$, then 
\begin{equation*}
\Gamma_k =\frac{D_k}{\mathcal{R}_0^2} (Ar + ((\mathcal{R}_0^{H_k})^2 D_k - \mathcal{R}_0^2)),
\end{equation*}
where $r = -\sum_{i=i}^{k-1}  ((\mathcal{R}_0^{H_i})^2 D_i - \mathcal{R}_0^2)$. Since the hosts $H_1, \ldots, H_{k-1}$ are suboptimal, we have $r >0$, hence $\Gamma_k$ is an increasing function of $A$.

If $D_k$ is large, then $D_1, \ldots, D_{k-1}$ are small and $(\mathcal{R}_0^{H_i})^2 D_i - \mathcal{R}_0^2 \approx - \mathcal{R}_0^2$ for  $i = 1, \ldots, k-1$. Therefore,  $r \approx (k-1)\mathcal{R}_0^2$. Furthermore, if   $(\mathcal{R}_0^{H_k})^2 D_k - \mathcal{R}_0^2 \approx 0 $ and  $A$ is large, then 
$$ \Gamma_k \approx  (k-1)A. $$


\begin{thebibliography}{99}


\bibitem{allan2009ecological} Allan, B. F., Langerhans, R. B., Ryberg, W. A., Landesman, W. J., Griffin, N. W., Katz, R. S., ... Chase, J. M. (2009). Ecological correlates of risk and incidence of west nile virus in the united states. \textit{Oecologia, 158} (4), 699?708.
\bibitem{epstein2006nipah} Epstein, J. H., Field, H. E., Luby, S., Pulliam, J. R. C.,  Daszak, P. (2006). Nipah virus: Impact, origins, and causes of emergence. \textit{Current Infectious Disease Reports , 8} (1), 59?65.
\bibitem{LoGiudice2008impact} LoGiudice, K., Duerr, S. T. K., Newhouse, M. J., Schmidt, K. A., Killilea, M. E.,  Ostfeld, R. S. (2008). Impact of host community composition on lyme disease risk. \textit{Ecology, 89}(10), 2841-2849. 
\bibitem{ostfeld2000biodiversity} Ostfeld, R. S.,  Keesing, F. (2000). Biodiversity and disease risk: the
case of lyme disease. \textit{Conservation Biology, 14} (3), 722?728.
\bibitem{keesing2009hosts} Keesing, F., Brunner, J., Duerr, S., Killilea, M., LoGiudice, K., Schmidt, K., . . . Ostfeld, R. S. (2009, 11). Hosts as ecological traps for the vector of lyme disease. \textit{Proceedings of the Royal Soci- ety B: Biological Sciences , 276} (1675), 3911?3919.

\bibitem{cronin2010host} Cronin, J. P., Welsh, M. E., Dekkers, M. G., Abercrombie, S. T., and Mitchell,
C. E. (2010). Host physiological phenotype explains pathogen reservoir
potential. Ecology Letters, 13 (10), 1221-1232.
\bibitem{dobson2004population} Dobson, A. (2004). Population dynamics of pathogens with multiple host
species. the american naturalist , 164 (S5), S64-S78.
\bibitem{johnson2010diversity} Johnson, P., and Thieltges, D. (2010). Diversity, decoys and the dilution effect: how ecological communities affect disease risk. Journal of Experimental
Biology, 213 (6), 961-970.
\bibitem{johnson2012parasite} Johnson, P. T., and Hoverman, J. T. (2012). Parasite diversity and coinfection
determine pathogen infection success and host tness. Proceedings of the
National Academy of Sciences, 109 (23), 9006-9011.
\bibitem{johnson2009community} Johnson, P. T., Lund, P. J., Hartson, R. B., and Yoshino, T. P. (2009). Community diversity reduces schistosoma mansoni transmission, host pathology
and human infection risk. Proceedings of the Royal Society of London B:
Biological Sciences, 276 (1662), 1657-1663.
\bibitem{keesing2010impacts} Keesing, Felicia, et al. "Impacts of biodiversity on the emergence and transmission of infectious diseases." Nature 468.7324 (2010): 647-652.
\bibitem{martin2006investment} Martin Ii, L. B., Hasselquist, D., and Wikelski, M. (2006). Investment in immune
defense is linked to pace of life in house sparrows. Oecologia, 147 (4),
565-575.
\bibitem{Murray2002}  Murray, J. D. (2002). Mathematical biology i. an introduction (3rd ed.,
Vol. 17). New York: Springer. doi: 10.1007/b98868
\bibitem{spivak1965calculus}  Spivak, M. (1965). Calculus on manifolds (Vol. 1). WA Benjamin New York.
\bibitem{swaddle2008increased} Swaddle, J. P., and Calos, S. E. (2008). Increased avian diversity is associated
with lower incidence of human west nile infection: observation of the
dilution effect. PloS one, 3 (6), e2488.
\bibitem{van2002reproduction} Van den Driessche, P., and Watmough, J. (2002). Reproduction numbers and
sub-threshold endemic equilibria for compartmental models of disease
transmission. Mathematical biosciences, 180 (1), 29-48.


\end{thebibliography}
\end{document}